\documentclass[11pt,a4paper]{article}
\pdfoutput=1

\usepackage[vmargin=33mm,hmargin=33mm]{geometry}
\usepackage{amsmath,amssymb,amsthm,booktabs,enumitem,graphicx,latexsym,url,xspace}
\usepackage[numbers,sort&compress]{natbib}
\usepackage[basic,small]{complexity}
\usepackage[small]{caption}
\usepackage[outercaption]{sidecap}
\usepackage{microtype}
\usepackage{hyperref}

\let\E=\undefined

\newcommand{\A}{\mathcal{A}}  
\newcommand{\F}{\mathcal{F}}  
\newcommand{\N}{\mathcal{N}}  
\DeclareMathOperator{\E}{E}   
\DeclareMathOperator{\opt}{opt}
\newcommand{\ca}{\mathsf{a}}
\newcommand{\cb}{\mathsf{b}}
\newcommand{\cab}{\{\ca,\cb\}}
\newcommand{\compl}[1]{-#1}

\newtheorem{theorem}{Theorem}

\newtheorem{lemma}[theorem]{Lemma}
\newtheorem{fact}[theorem]{Fact}

\theoremstyle{definition}

\theoremstyle{remark}
\newtheorem{remark}{Remark}

\hypersetup{
    colorlinks=true,
    linkcolor=black,
    citecolor=black,
    filecolor=black,
    urlcolor=[rgb]{0,0.1,0.5},
    pdfauthor={Juho Hirvonen, Joel Rybicki, Stefan Schmid, Jukka Suomela},
    pdftitle={Large Cuts with Local Algorithms on Triangle-Free Graphs},
}

\begin{document}

\bigskip
\bigskip
\bigskip
\begin{raggedright}
{\LARGE \textbf{Large Cuts with Local Algorithms on \\ Triangle-Free Graphs}\par}
\bigskip
\bigskip
\bigskip

\textbf{Juho Hirvonen}

\smallskip
{\small Helsinki Institute for Information Technology HIIT,\\Department of Information and Computer Science, Aalto University, Finland\\
\nolinkurl{juho.hirvonen@aalto.fi}\par}
\bigskip

\textbf{Joel Rybicki}

\smallskip
{\small Helsinki Institute for Information Technology HIIT,\\Department of Information and Computer Science, Aalto University, Finland\\
\nolinkurl{joel.rybicki@aalto.fi}\par}
\bigskip

\textbf{Stefan Schmid}

\smallskip
{\small TU Berlin \& T-Labs, Germany\\
\nolinkurl{stefan@net.t-labs.tu-berlin.de}\par}
\bigskip

\textbf{Jukka Suomela}

\smallskip
{\small Helsinki Institute for Information Technology HIIT,\\Department of Information and Computer Science, Aalto University, Finland\\
\nolinkurl{jukka.suomela@aalto.fi}\par}
\end{raggedright}

\bigskip
\bigskip
\bigskip
\noindent\textbf{Abstract.}
We study the problem of finding \emph{large cuts in $d$-regular triangle-free graphs}. In prior work, Shearer (1992) gives a randomised algorithm that finds a cut of expected size $(1/2 + 0.177/\sqrt{d})m$, where $m$ is the number of edges. We give a simpler algorithm that does much better: it finds a cut of expected size $(1/2 + 0.28125/\sqrt{d})m$. As a corollary, this shows that in any $d$-regular triangle-free graph there exists a cut of at least this size.

Our algorithm can be interpreted as a very efficient \emph{randomised distributed algorithm}: each node needs to produce only one random bit, and the algorithm runs in one synchronous communication round. This work is also a case study of applying \emph{computational techniques} in the design of distributed algorithms: our algorithm was designed by a computer program that searched for optimal algorithms for small values of~$d$.

\thispagestyle{empty}
\setcounter{page}{0}
\newpage

\section{Introduction}

We study the problem of finding \emph{large cuts} in \emph{triangle-free graphs}. In particular, we are interested in the design of \emph{fast and simple randomised distributed algorithms}.

\subsection{Random Cuts}\label{ssec:intro-cut}

Let $G = (V,E)$ be a simple undirected graph. A \emph{cut} is a function $c\colon V \to \cab$ that labels the nodes with symbols $\ca$ and $\cb$. An edge $\{u,v\} \in E$ is a \emph{cut edge} if $c(u) \ne c(v)$. We use the convention that the \emph{weight} $w(c)$ of a cut $c$ is the fraction of edges that are cut edges; that is, the weight of the cut is normalised so that it is in the range $[0,1]$. See Figure~\ref{fig:cut} for an illustration.

\begin{figure}[h]
    \centering
    \includegraphics[page=1]{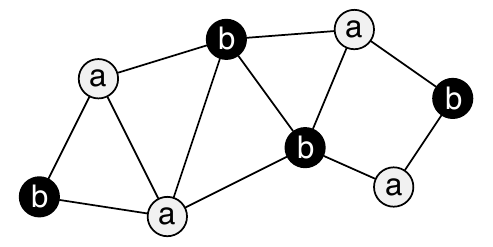}
    \caption{A cut $c\colon V \to \cab$ of weight $w(c) = \frac{10}{12}$.}\label{fig:cut}
\end{figure}

While the problem of finding a maximum cut (or a good approximation of one) is NP-hard~\cite{garey79computers,papadimitriou91optimization,hastad01inapproximability,trevisan00gadgets,khot07maxcut}, there is a very simple randomised algorithm that finds a relatively large cut: for each node $v$, pick $c(v) \in \cab$ independently and uniformly at random. We say that $c$ is a \emph{uniform random cut}.

In a uniform random cut, each edge is a cut edge with probability $1/2$. It follows that the expected weight of a uniform random cut is also $1/2$.

\subsection{Regular Triangle-Free Graphs}

In general graphs, we cannot expect to find cuts that are much better than uniform random cuts. For example, in a complete graph on $n$ nodes, the weight of any cut is at most $1/2 + O(1/n)$.

However, there is a family of graphs that makes for a much more interesting case from the perspective of the max-cut problem: regular triangle-free graphs. Erd\H{o}s~\cite{erdos79problemsresults} raised the problem of estimating the minimum possible size of a maximum cut in a high-girth graph, and especially the case of triangle-free graphs attracted much interest from the research community~\cite{shearer92bipartite,poljak95maximum,alon96bipartite}.

Accordingly, from now on, we assume that $G$ is a $d$-regular graph for some constant $d \ge 2$, and that there are no triangles (cycles of length three) in~$G$. While focusing on regular triangle-free graphs may seem overly restrictive, our algorithm can be applied in a much more general setting; we will briefly discuss extensions in Section~\ref{sec:concl}.

\pagebreak

\subsection{Shearer's Algorithm}\label{ssec:intro-shearer}

In triangle-free graphs, it is easy to find cuts that are (in expectation) larger than uniform random cuts. Nevertheless, a uniform random cut is a good starting point.

Shearer's \cite{shearer92bipartite} algorithm proceeds as follows. Pick three uniform random cuts $c_1$, $c_2$, and $c_3$. For each node $v$, let
\[
    \ell(v) = \bigl|\{v,u\} \in E : c_1(v) = c_1(u) \}\bigr|
\]
be the number of like-minded neighbours in $c_1$. Then the output of a node $v$ is
\begin{equation}
    c(v) = \begin{cases}
        c_1(v), & \text{if } \ell(v) < d/2, \\
        c_1(v), & \text{if } \ell(v) = d/2 \text{ and } c_3(v) = 0, \\
        c_2(v), & \text{if } \ell(v) = d/2 \text{ and } c_3(v) = 1, \\
        c_2(v), & \text{if } \ell(v) > d/2.
    \end{cases}
    \label{eq:salg}
\end{equation}
Put otherwise, a node follows $c_1$ if it seems that there are many cut edges w.r.t.\ $c_1$ in its immediate neighbourhood, and it falls back to another cut $c_2$ otherwise. The value $c_3(v)$ is just used as a random tie-breaker.

Shearer \cite{shearer92bipartite} shows that the expected weight of cut \eqref{eq:salg} is at least
\begin{equation}
    \frac{1}{2} + \frac{\sqrt{2}}{8 \sqrt{d}} \,\approx\, \frac{1}{2} + \frac{0.177}{\sqrt{d}}
    \label{eq:sperf}
\end{equation}
in $d$-regular triangle-free graphs.

\subsection{Our Algorithm}\label{ssec:intro-alg}

Shearer's algorithm can be characterised as follows: take a uniform random cut $c_1$ and then improve it with the help of a \emph{randomised} rule described in \eqref{eq:salg}. In this work, we show that we can do much better with the help of a simple \emph{deterministic} rule.

In our algorithm we pick one uniform random cut $c_1$. Again, each node $v$ counts the number of like-minded neighbours
\[
    \ell(v) = \bigl|\{v,u\} \in E : c_1(v) \ne c_1(u) \}\bigr|.
\]
We define the threshold
\begin{equation}
    \tau = \biggl\lceil \frac{d + \sqrt{d}}{2} \biggr\rceil.
    \label{eq:tau}
\end{equation}
Now the output of a node $v$ is simply
\begin{equation}
    c(v) = \begin{cases}
        c_1(v), &\text{ if } \ell(v) < \tau, \\
        \compl{c_1(v)}, &\text{ if } \ell(v) \ge \tau.
    \end{cases}
    \label{eq:alg}
\end{equation}
Here $\compl{c_1(v)}$ is the complement of $c_1(v)$, that is, $\compl{\ca} = \cb$ and $\compl{\cb} = \ca$. In the algorithm each node simply changes its mind if it seems that there are too many like-minded neighbours.

It is not obvious that such a rule makes sense, or that this particular choice of $\tau$ is good. Nevertheless, we show in this work that the expected weight of cut \eqref{eq:alg} is at least
\begin{equation}
    \frac{1}{2} + \frac{9}{32 \sqrt{d}} \,=\, \frac{1}{2} + \frac{0.28125}{\sqrt{d}},
    \label{eq:perf}
\end{equation}
which is much larger than Shearer's bound \eqref{eq:sperf}, at least in low-degree graphs. As a corollary, any $d$-regular triangle-free graph admits a cut of at least this size.

Our algorithm can be implemented very efficiently in a \emph{distributed} setting: each node only needs to produce one random bit, and the algorithm only requires one communication round. In Shearer's algorithm each node has to produce up to three random bits.

Perhaps the most interesting feature of the algorithm is that it was not designed by a human being---it was discovered by a computer program. Indeed, cuts in triangle-free graphs serve as an example of a computational problem in which \emph{computer-aided methods} can be used to partially \emph{automate algorithm design and analysis} (this process is also known as ``algorithm synthesis'' or ``protocol synthesis''). There is a wide range of other graph problems in which a similar approach has a lot of potential as a shortcut to the discovery of new distributed algorithms.

In Section~\ref{sec:design}, we outline the procedure that we used to design the algorithm, and then present an analysis of its performance. In Section~\ref{sec:concl} we discuss how to apply the algorithm in a more general setting beyond regular triangle-free graphs.

\section{Algorithm Design and Analysis}\label{sec:design}

We begin this section with an informal overview of so-called neighbourhood graphs. The formal definitions that we use in this work are given after that.

\subsection{Neighbourhood Graphs in Prior Work}\label{ssec:neighg-prior}

In the context of distributed systems, the \emph{radius-$t$ neighbourhood} $N(t,v)$ of a node $v$ refers to all information that node $v$ may gather in $t$ communication rounds. Depending on the model of computation that we use, this may include all nodes that are within distance $t$ from $v$, the edges incident to these nodes, their local inputs, and the random bits that these nodes have generated. The idea is that whatever decision node $v$ takes, it can only depend on its radius-$t$ neighbourhood---any distributed algorithm $\A$ that runs in $t$ communication rounds can be interpreted as a mapping from local neighbourhoods to local outputs.

A \emph{neighbourhood graph} $\N_t$ is a graph representation of all possible radius-$t$ neighbourhoods that a distributed algorithm may encounter. Each node $N \in V(\N_t)$ of the neighbourhood graph corresponds to a possible local neighbourhood: there is at least one communication network in which some node has a local neighbourhood isomorphic to $N$. We have an edge $\{N_1,N_2\} \in E(\N_t)$ in the neighbourhood graph if there is some communication network in which nodes with local neighbourhoods $N_1$ and $N_2$ are adjacent; see Figure~\ref{fig:ng} for an example.

\begin{figure}
    \centering
    \includegraphics[page=2]{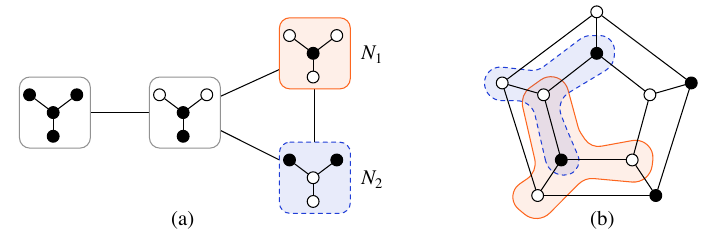}
    \caption{In this example, we study the family $\F$ of 3-regular triangle-free graphs that are labelled with two colours, black and white. (a)~A small part of neighbourhood graph $\N_t$ for $t = 1$. (b)~There exists a graph $G \in \F$ in which local neighbourhoods $N_1$ and $N_2$ are adjacent; hence nodes $N_1$ and $N_2$ are adjacent in the neighbourhood graph.}\label{fig:ng}
\end{figure}

Neighbourhood graphs are a convenient concept in the study of graph colouring algorithms, both from the perspective of traditional algorithm design \cite{linial92locality,naor91lower,kelsen96neighborhood,kuhn06complexity,fraigniaud07distributed} and from the perspective of computational algorithm design \cite{rybicki11msc}. The key observation is that the following two statements are equivalent:
\begin{itemize}[noitemsep]
    \item $\A \colon V(\N_t) \to \{1,2,\dotsc,k\}$ is a proper colouring of the neighbourhood graph $\N_t$,
    \item $\A$ is a distributed algorithm that finds a proper $k$-colouring in $t$ rounds.
\end{itemize}
To see this, consider any graph $G$. If nodes $u$ and $v$ are adjacent in $G$, then their local views $N(t,u)$ and $N(t,v)$ are adjacent in $\N_t$, and by assumption $\A$ assigns a different colour to $N(t,u)$ and $N(t,v)$. Hence distributed algorithm $\A$ finds a proper $k$-colouring of $G$. Conversely, if algorithm $\A$ finds a proper colouring in any communication network, it defines a proper $k$-colouring of $\N_t$.

In summary, colourings of the neighbourhood graph correspond to distributed algorithms for graph colouring, and vice versa. In general, a similar property does \emph{not} hold for arbitrary graph problems. For example, there is no one-to-one correspondence between maximal independent sets of $\N_t$ and distributed algorithms that find maximal independent sets~\cite[Section 8.5]{rybicki11msc}.

However, as we will see in this work, we can use neighbourhood graphs also in the context of the maximum cut problem. It turns out that we can define a \emph{weighted} version of neighbourhood graphs, so that there is a one-to-one correspondence between \emph{heavy} cuts in the weighted neighbourhood graph, and randomised distributed algorithms that find \emph{large} cuts in expectation.

\subsection{Model of Distributed Computing}\label{ssec:model}

Next, we formalise the model of distributed computing that is sufficient for the purposes of our algorithm. Fix the parameter $d$; recall that we are interested in $d$-regular triangle-free graphs. Let $G = (V,E)$ be such a graph, and let $c$ be a uniform random cut in $G$. The \emph{local neighbourhood} of a node $v$ is $N_c(v) = (c(v), \ell_c(v))$, where
\[
    \ell_c(v) = \bigl|\{v,u\} \in E : c(v) = c(u) \}\bigr|
\]
is the number of neighbours with the same random bit. Note that there are only $2d+2$ possible local neighbourhoods.

A distributed algorithm is a function $\A$ that associates an output $\A(N) \in \cab$ with each local neighbourhood $N$. For any $d$-regular triangle-free graph $G = (V,E)$, function $\A$ defines a randomised process that produces a random cut $c'$ as follows:
\begin{enumerate}[noitemsep]
    \item Pick a uniform random cut $c$.
    \item For each node $v$, let $c'(v) = \A(N_c(v))$.
\end{enumerate}
We use the notation $\A(G)$ for the random cut $c'$ produced by algorithm $\A$ in graph $G$. In particular, we are interested in the quantity $\E[w(\A(G))]$, the expected weight of cut~$c'$.

A priori, we might expect that $\E[w(\A(G))]$ would depend on $G$. However, as we will soon see, this is not the case---it only depends on parameter $d$ and algorithm $\A$.

\subsection{Weighted Neighbourhood Graph}

A \emph{weighted digraph} is a pair $D = (V,w)$ with $w \colon V \times V \to [0,\infty)$. Here $V$ is the set of nodes, and $w$ associates a non-negative \emph{weight} $w(x,y) \ge 0$ with each directed edge $(x,y) \in V \times V$. Let $c\colon V \to \cab$ be a cut in weighted digraph $D$. The weight of cut $c$ is
\[
    \label{eq:w}
    w(c) = \sum_{\substack{(u,v) \in V \times V, \\ c(u) \ne c(v)}} w(u,v),
\]
the total weight of all cut edges.

The \emph{weighted neighbourhood graph} $\N = (V_\N,w_\N)$ is a weighted digraph defined as follows (see Figure~\ref{fig:wng} for an illustration). The set of nodes
\[
    V_\N = \bigl\{ (k, i) : k \in \cab,\, i \in \{0,1,\dotsc,d\} \bigr\}
\]
consists of all possible neighbourhoods that we may encounter in $d$-regular triangle-free graphs. We define the edge weights as follows:
\[
    w_\N\bigl((k_1,i_1), (k_2,i_2)\bigr) = \begin{cases}
        \displaystyle
        \frac{1}{4^d}\binom{d-1}{i_1}\binom{d-1}{i_2} & \text{if } k_1 \ne k_2, \\[3.5ex]
        \displaystyle
        \frac{1}{4^d}\binom{d-1}{i_1-1}\binom{d-1}{i_2-1} & \text{if } k_1 = k_2.
    \end{cases}
\]
We follow the convention that $\binom{n}{k} = 0$ for $k < 0$ and $k > n$.

\begin{figure}[t]
    \centering
    \includegraphics[page=3]{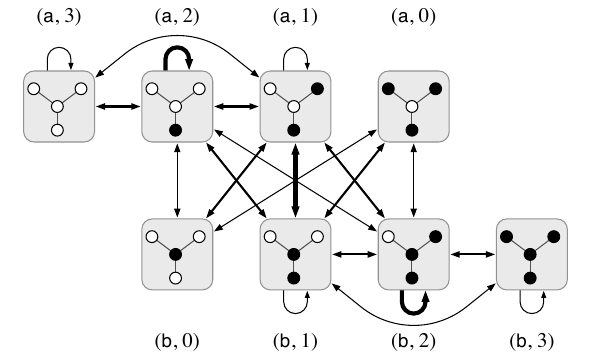}
    \caption{Weighted neighbourhood graph $\N$ for $d = 3$. Edge weights are denoted by line widths; missing edges have weight~$0$.
    Note that the digraph is symmetric; however, we prefer the directed representation so that we do not need special treatment for self-loops.}\label{fig:wng}
\end{figure}

Note that the weights are symmetric, and the total weight of all edges is $1$. The following lemma shows that the weight of the edge $(N_1, N_2)$ in the neighbourhood graph equals the probability of ``observing'' adjacent neighbourhoods of types $N_1$ and $N_2$; see Figure~\ref{fig:wng2}. Note that the probability does not depend on the choice of graph $G$ or edge $\{u,v\}$.

\begin{figure}[t]
    \centering
    \includegraphics[page=4]{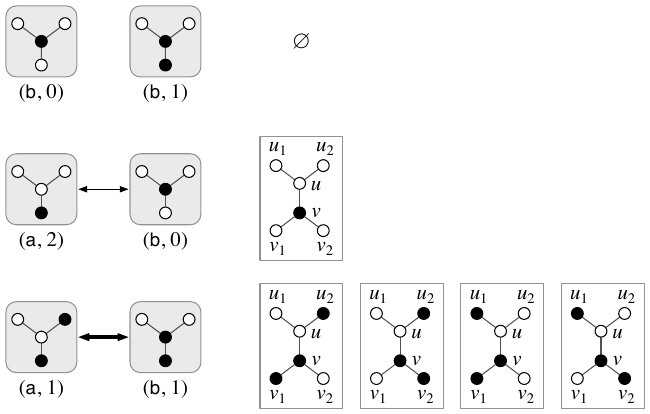}
    \caption{Selected examples of edge weights in the weighted neighbourhood graph $\N$ (see Figure~\ref{fig:wng}). We have $w_\N((\cb,0), (\cb,1)) = 0$, $w_\N((\ca,2), (\cb,0)) = 1/64$, and $w_\N((\ca,1), (\cb,1)) = 1/16$. It is not possible to have a graph in which we have adjacent neighbourhoods of types $(\cb,0)$ and $(\cb,1)$. Adjacent neighbourhoods of types $(\ca,2)$ and $(\cb,0)$ are fairly rare, while adjacent neighbourhoods of types $(\ca,1)$ and $(\cb,1)$ are much more common.}\label{fig:wng2}
\end{figure}

\begin{lemma}\label{lem:ngraph-weight}
    Let $G$ be a $d$-regular triangle-free graph, and let $\{u,v\}$ be an edge of~$G$. Consider a uniform random cut $c$ of $G$. Then for any given neighbourhoods $N_1, N_2 \in V_\N$ we have
    \[
        \Pr\bigl[N_c(u) = N_1 \text{ and } N_c(v) = N_2\bigr]
        \,=\, w_\N(N_1, N_2).
    \]
\end{lemma}
\begin{proof}
    In what follows, we will denote the neighbours of $u$ by $u_1, u_2, \dotsc, u_d$ where $u_d = v$. Similarly, the neighbours of $v$ are $v_1, v_2, \dotsc, v_d$ where $v_d = u$. As $G$ is triangle-free, sets $S_u = \{u_1, u_2, \dots, u_{d-1}\}$ and $S_v = \{v_1, v_2, \dots, v_{d-1}\}$ are disjoint. In particular, the random variables $c(x)$ for $x \in S_u \cup S_v$ are independent.

    Let $N_1 = (k_1, i_1)$ and $N_2 = (k_2, i_2)$. There are two cases. First assume that $k_1 = k_2$. Then
    \[
    \begin{split}
        \Pr\bigl[N_c(u) = N_1 \text{ and } N_c(v) = N_2\bigr]
        \ = \ &\Pr\bigl[c(u) = k_1 \text{ and } c(v) = k_2\bigr] \cdot \\
              &\Pr\bigl[|\{y \in S_u : c(y) = k_1\}| = i_1 - 1\bigr] \cdot \\
              &\Pr\bigl[|\{y \in S_v : c(y) = k_2\}| = i_2 - 1\bigr] \\
        \ = \ &\frac{1}{4}
               \cdot \frac{1}{2^{d-1}} \binom{d-1}{i_1 - 1}
               \cdot \frac{1}{2^{d-1}} \binom{d-1}{i_2 - 1} \\
        \ = \ &w_\N(N_1, N_2).
    \end{split}
    \]
    Second, assume that $k_1 \ne k_2$. Then
    \[
    \begin{split}
        \Pr\bigl[N_c(u) = N_1 \text{ and } N_c(v) = N_2\bigr]
        \ = \ &\Pr\bigl[c(u) = k_1 \text{ and } c(v) = k_2\bigr] \cdot \\
              &\Pr\bigl[|\{y \in S_u : c(y) = k_1\}| = i_1\bigr] \cdot \\
              &\Pr\bigl[|\{y \in S_v : c(y) = k_2\}| = i_2\bigr] \\
        \ = \ &\frac{1}{4}
               \cdot \frac{1}{2^{d-1}} \binom{d-1}{i_1}
               \cdot \frac{1}{2^{d-1}} \binom{d-1}{i_2} \\
        \ = \ &w_\N(N_1, N_2). \qedhere
    \end{split}
    \]
\end{proof}

\pagebreak

\subsection{Cuts in Neighbourhood Graphs}

Any function $\A \colon V_\N \to \cab$ can be interpreted in two ways:
\begin{enumerate}
    \item A cut of weight $w_\N(\A)$ in the weighted neighbourhood graph $\N$.
    \item A distributed algorithm that finds a cut in any $d$-regular triangle-free graph: the algorithm picks a uniform random cut $c$, and then node $v$ outputs $\A(N_c(v))$.
\end{enumerate}
The following lemma shows that the two interpretations are closely related: if $\A$ is a cut of weight $w$ in neighbourhood graph $\N$, then it immediately gives us a distributed algorithm that finds a cut of \emph{expected} weight $w$ in any $d$-regular triangle-free graph.

\begin{lemma}\label{lem:cut-is-alg}
    If $\A \colon V_\N \to \cab$ is a cut in neighbourhood graph $\N$, and $G$ is a $d$-regular triangle-free graph, then $\E[w(\A(G))] = w_\N(\A)$.
\end{lemma}
\begin{proof}
    Fix a graph $G$ and an edge $\{u,v\}$ of $G$. By Lemma~\ref{lem:ngraph-weight} we have
    \[
    \begin{split}
        w_\N(\A)
        &= \sum_{\A(N_1) \ne \A(N_2)} w_\N(N_1,N_2) \\
        &= \sum_{\A(N_1) \ne \A(N_2)} \Pr\bigl[N_c(u) = N_1 \text{ and } N_c(v) = N_2\bigr] \\
        &= \Pr\bigl[\A(N_c(u)) \ne \A(N_c(v))\bigr].
    \end{split}
    \]
    The claim follows by summing over all edges $\{u,v\}$ of $G$.
\end{proof}

\subsection{Computational Algorithm Design}\label{ssec:compalg}

Now we have all the tools that we need. Lemma~\ref{lem:cut-is-alg} gives a one-to-one correspondence between large cuts of the neighbourhood graph and distributed algorithms that find large cuts. For any fixed value of $d$, the task of designing a distributed algorithm is now straightforward:
\begin{enumerate}[noitemsep]
    \item Construct the weighted neighbourhood graph $\N$.
    \item Find a heavy cut in $\N$.
\end{enumerate}
See Figure~\ref{fig:wng-cut} for an example. For $d = 3$, the heaviest cut $\A_{\opt}$ of $\N$ is
\begin{equation}
    \label{eq:alg3}
    \A_{\opt}((k,i)) = \begin{cases}
        k & \text{if } i < 3, \\
        \compl{k} & \text{if } i \ge 3.
    \end{cases}
\end{equation}
This is also the best possible algorithm for this value of $d$, for the model of computing that we defined in Section~\ref{ssec:model}.

\begin{figure}
    \centering
    \includegraphics[page=5]{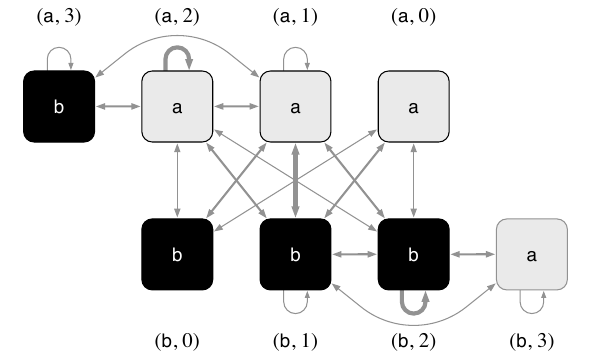}
    \caption{Maximum-weight cut in the weighted neighbourhood graph for $d = 3$.}\label{fig:wng-cut}
\end{figure}

\begin{remark}
    The reader may want to compare \eqref{eq:alg3} with Section~\ref{ssec:intro-alg}. For $d = 3$, the algorithms are identical, albeit with a slightly different notation. Note that $\tau_3 = 3$.
\end{remark}

Of course finding a maximum-weight cut is hard in the general case. However, in this particular case neighbourhood graphs are relatively small (only $2d+2$ nodes).

\pagebreak

While the smallest cases could be easily solved with brute force, slightly more refined approaches are helpful for moderate values of $d$. We took the following approach. First, we reduced the max-weight-cut instance $\N$ to a max-weight-SAT instance $\phi$ in a straightforward manner:
\begin{itemize}
    \item For each node $u \in V_\N$ we have a Boolean variable $x_u$ in formula $\phi$.
    \item For each edge $(u,v)$ of weight $w_\N(u,v)$ we have two clauses in formula $\phi$, both of weight $w_\N(u,v)$:
    \[
        x_u \lor x_v
        \quad \text{and} \quad
        \neg x_u \lor \neg x_v
    \]
    Note that at least one of these clauses is always satisfied, while both of them are satisfied if and only if $x_u$ and $x_v$ have different values.
\end{itemize}
Now it is easy to see that a variable assignment $x$ of $\phi$ that maximises the total weight of satisfied clauses also gives a maximum-weight cut $\A$ in $\N$: let $\A(u) = \ca$ iff $x_u$ is true. More precisely, the total weight of the clauses satisfied by $x$ is $W + w_\N(\A)$, where $W$ is the total weight of all edges.

With this reduction, we can then resort to off-the-self max-weight-SAT solvers. In our experiments we used \emph{akmaxsat} solver \cite{kugel10akmaxsat}; with it we can solve the cases $d = 2, 3, \dotsc, 32$ very quickly (e.g., the case $d = 32$ on a low-end laptop in less than 5 seconds).

Surprisingly, in all cases the max-weight cut has the following simple structure:
\begin{equation}
    \label{eq:alggen}
    \A_\tau((k,i)) = \begin{cases}
        k & \text{if } i < \tau, \\
        \compl{k} & \text{if } i \ge \tau.
    \end{cases}
\end{equation}
The exact values of $\tau$ for the heaviest cuts are given in Table~\ref{tab:smalld}; note that all values are slightly larger than $d/2$.

\begin{table}
    \centering
    {\small\begin{tabular}{@{}rr@{ }r@{ }r@{ }r@{ }r@{ }r@{ }r@{ }r@{ }r@{ }r@{ }r@{ }r@{ }r@{ }r@{ }r@{ }r@{ }r@{ }r@{ }r@{ }r@{ }r@{ }r@{ }r@{ }r@{ }r@{ }r@{ }r@{ }r@{ }r@{ }r@{ }r@{}}
$d$:
& 2 & 3 & 4 & 5 & 6 & 7 & 8 & 9 & 10 & 11 & 12 & 13 & 14 & 15 & 16 & 17 & 18 & 19 & 20 & 21 & 22 & 23 & 24 & 25 & 26 & 27 & 28 & 29 & 30 & 31 & 32 \\
$\tau_{\opt}$:
& 2 & 3 & 3 & 4 & 5 & 5 & 6 & 6 & 7 & 7 & 8 & 9 & 9 & 10 & 10 & 11 & 11 & 12 & 12 & 13 & 14 & 14 & 15 & 15 & 16 & 16 & 17 & 17 & 18 & 18 & 19 \\
\end{tabular}
}
    \caption{Optimal threshold $\tau_{\opt}$ for small values of $d$.}\label{tab:smalld}
\end{table}

\subsection{Generalisation}\label{ssec:gen}

Now it is easy to generalise the findings: we can make the educated guess that algorithms of form \eqref{eq:alggen} are good also in the case of a general~$d$. All we need to do is to find a general expression for the threshold $\tau$, and prove that algorithm $\A_\tau$ indeed works well in the general case.

To facilitate algorithm analysis, let us define the shorthand notation
\[
    \alpha(\tau, d) = w_\N(\A_\tau) = \E[w(\A_\tau(G))]
\]
for the performance of algorithm $\A_\tau$. It is easy to see that
$\alpha(0,d) = \alpha(d+1,d) = 1/2$,
as the threshold value of $\tau = d+1$ simply means that algorithm $\A_\tau$ outputs a uniform random cut, while $\tau = 0$ means that $\A_\tau$ outputs the complement of the uniform random cut. The general shape of $\alpha(\tau,d)$ is illustrated in Figure~\ref{fig:alpha}.

\begin{figure}
    \centering
    \includegraphics{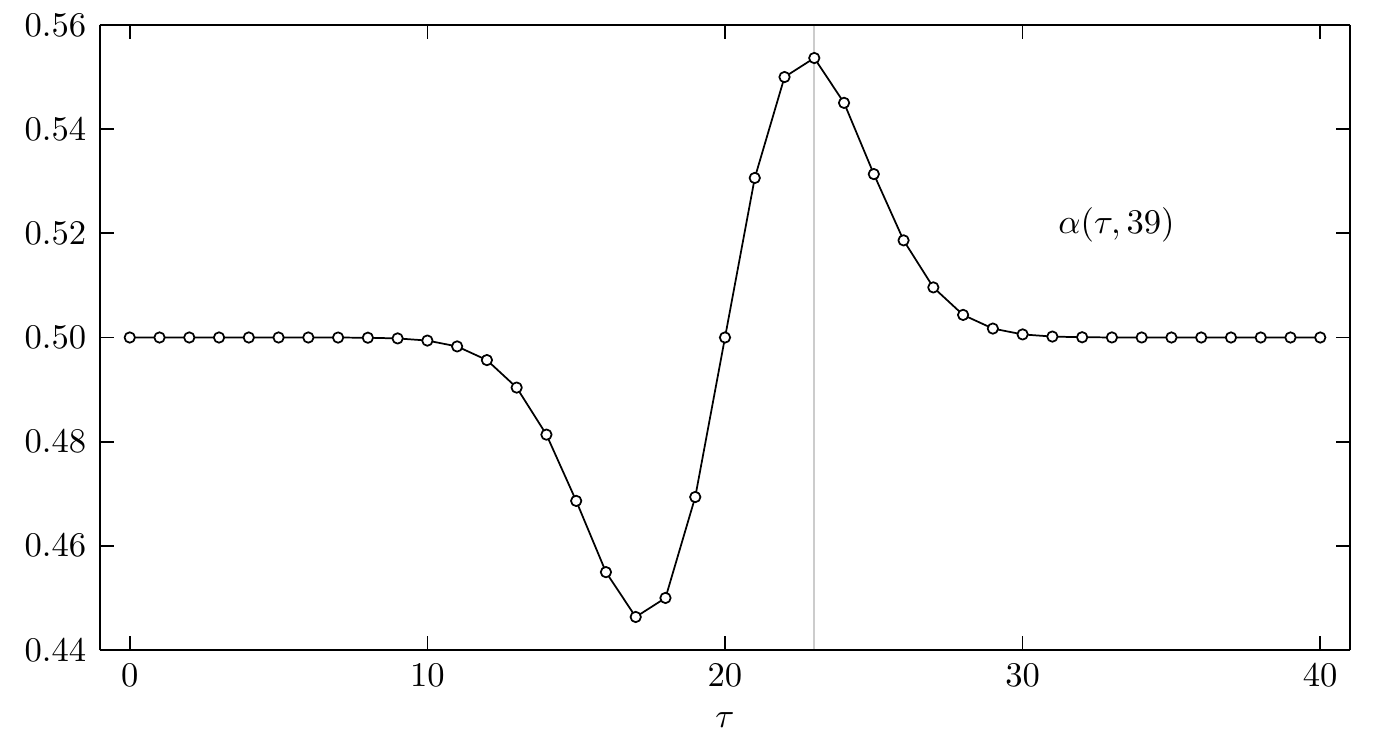}
    \caption{$\alpha(\tau,39)$ for $\tau = 0, 1, \dotsc, 40$.}\label{fig:alpha}
\end{figure}

We are interested in the region $\tau > d/2$, where $\alpha(\tau,d) \ge 1/2$. In the following, we derive a relatively simple expression for $\alpha(\tau,d)$ in this region---the proof strategy is inspired by Shearer~\cite{shearer92bipartite}.

\begin{lemma}\label{lem:wtau}
    For all $d$ and $\tau > d/2$ we have
    \[
        \alpha(\tau, d)
        = \frac{1}{2} + \frac{1}{4^{d-1}}
          \binom{d-1}{\tau-1}
          \sum_{i=d-\tau+1}^{\tau-1} \binom{d-1}{i}.
    \]
\end{lemma}
\begin{proof}
    Fix a triangle-free $d$-regular graph $G = (V,E)$. Recall that $c$ is a uniform random cut, $N_c(v) = (c(v), \ell_c(v))$ is the local neighbourhood of node $v \in V$, and $\A(N_c(v))$ is the output of algorithm $\A$ at node $v \in V$.

    Consider an edge $\{u,v\}$ of $G$. We will calculate the probability that $e$ is a cut edge. To this end, define
    \begin{align*}
        p &= \Pr\bigl[c(u) \ne c(v) \text{ and } \ell_c(u), \ell_c(v) \ge \tau\bigr], \\
        q &= \Pr\bigl[c(u) \ne c(v) \text{ and } \ell_c(u), \ell_c(v) < \tau\bigr], \\
        r &= \Pr\bigl[c(u) =   c(v) \text{ and either } \ell_c(u) < \tau \le \ell_c(v) \text{ or } \ell_c(v) < \tau \le \ell_c(u)\bigr].
    \end{align*}
    These are precisely the cases in which $\A(N_c(u)) \ne \A(N_c(v))$; hence $\{u,v\}$ is a cut edge with probability $p+q+r$. For each $x \in \{u,v\}$, let
    \begin{align*}
        p_x &= \Pr\bigl[\ell_c(x) \ge \tau \mid c(u) \ne c(v)\bigr], \\
        q_x &= \Pr\bigl[\ell_c(x) <   \tau \mid c(u) \ne c(v)\bigr], \\
        r_x &= \Pr\bigl[\ell_c(x) \ge \tau \mid c(u) =   c(v)\bigr].
    \end{align*}
    Now we have the following identities:
    \begin{align*}
        p &= \frac{1}{2} p_u p_v, &
        q &= \frac{1}{2} q_u q_v, &
        r &= \frac{1}{2} (r_v (1-r_u) + r_u(1-r_v)).
    \end{align*}
    By definition, $q_x = 1-p_x$, and by symmetry, $p_u = p_v$, $q_u = q_v$, and $r_u = r_v$. Hence the probability that $\{u,v\}$ is a cut edge is
    \begin{equation}\label{eq:wtau-pqr}
    \begin{split}
        p + q + r
        &= \frac{1}{2} p_u^2 + \frac{1}{2} q_u^2 + r_u(1-r_u)
        = \frac{1}{2} + p_u(p_u - 1) + r_u(1 - r_u) \\
        &= \frac{1}{2} - p_u q_u + r_u(p_u + q_u - r_u)
        = \frac{1}{2} + (r_u - p_u) (q_u - r_u).
    \end{split}
    \end{equation}

    An argument similar to what we used in Lemma~\ref{lem:ngraph-weight} gives
    \begin{align*}
        p_u &= \frac{1}{2^{d-1}} \sum_{i=\tau}^{d-1} \binom{d-1}{i}, &
        q_u &= \frac{1}{2^{d-1}} \sum_{i=0}^{\tau-1} \binom{d-1}{i}, &
        r_u &= \frac{1}{2^{d-1}} \sum_{i=\tau-1}^{d-1} \binom{d-1}{i}.
    \end{align*}
    Recall that we assumed that $\tau > d/2$; hence $\tau - 1 \ge d - \tau$ and
    \begin{align*}
        2^{d-1}(r_u - p_u) &= \binom{d-1}{\tau-1}, \\
        2^{d-1}(q_u - r_u) &= \sum_{i=0}^{\tau-1} \binom{d-1}{i}
                            - \sum_{i=0}^{d-\tau} \binom{d-1}{i}
                            = \sum_{i=d-\tau+1}^{\tau-1} \binom{d-1}{i}.
    \end{align*}
    From \eqref{eq:wtau-pqr} we therefore obtain
    \[
        p + q + r
        = \frac{1}{2} + \frac{1}{4^{d-1}} \binom{d-1}{\tau-1} \sum_{i=d-\tau+1}^{\tau-1} \binom{d-1}{i}.
        \qedhere
    \]
\end{proof}

Now we can easily find an optimal threshold $\tau$ for any given $d$: simply try all $d/2$ possible values and apply Lemma~\ref{lem:wtau}. Figure~\ref{fig:tau} is a plot of optimal $\tau$ for $d = 2, 3, \dotsc, 1000$. At least for small values of $d$, it appears that
\[
    \tau \approx \frac{d+1}{2} + 0.439 \sqrt{d}
\]
is close to the optimum. For notational convenience, we pick a slightly larger value
\[
    \tau = \biggl\lceil \frac{d + \sqrt{d}}{2} \biggr\rceil.
\]
Now we have arrived at the algorithm that we already described in Section~\ref{ssec:intro-alg}.

\begin{figure}[p]
    \centering
    \includegraphics{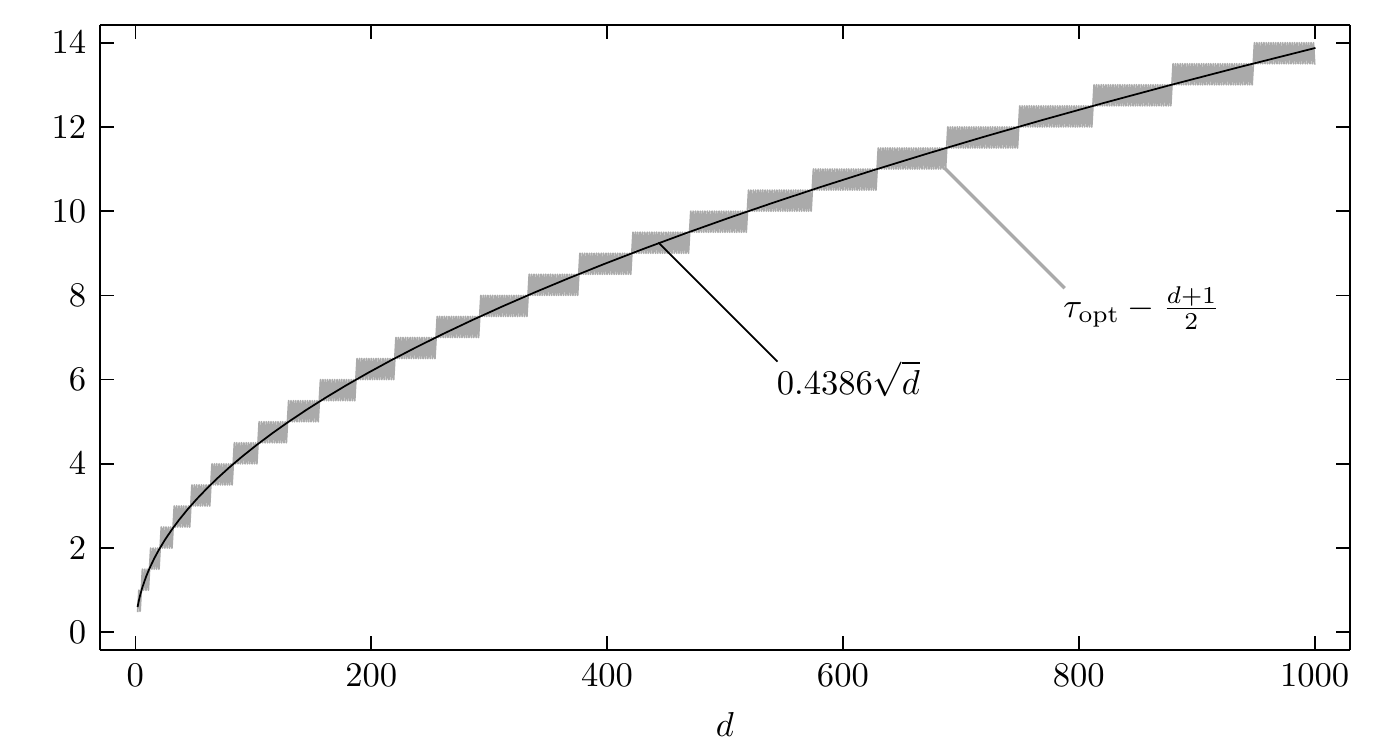}
    \caption{Optimal threshold $\tau_{\opt}$ for $d = 2, 3, \dotsc, 1000$.}\label{fig:tau}
\end{figure}

\pagebreak

What remains is a proof of the performance guarantee \eqref{eq:perf}. Figure~\ref{fig:small} gives some intuition on how good the bounds are.

\begin{theorem}\label{thm:main}
    Let $d \ge 2$ and
    \[
        \tau = \biggl\lceil \frac{d + \sqrt{d}}{2} \biggr\rceil.
    \]
    Then
    \[
        \alpha(\tau, d) \ge \frac{1}{2} + \frac{9}{32 \sqrt{d}}.
    \]
\end{theorem}
\begin{proof}
    See Appendix~\ref{app:main}.
\end{proof}

\begin{figure}[p]
    \centering
    \includegraphics{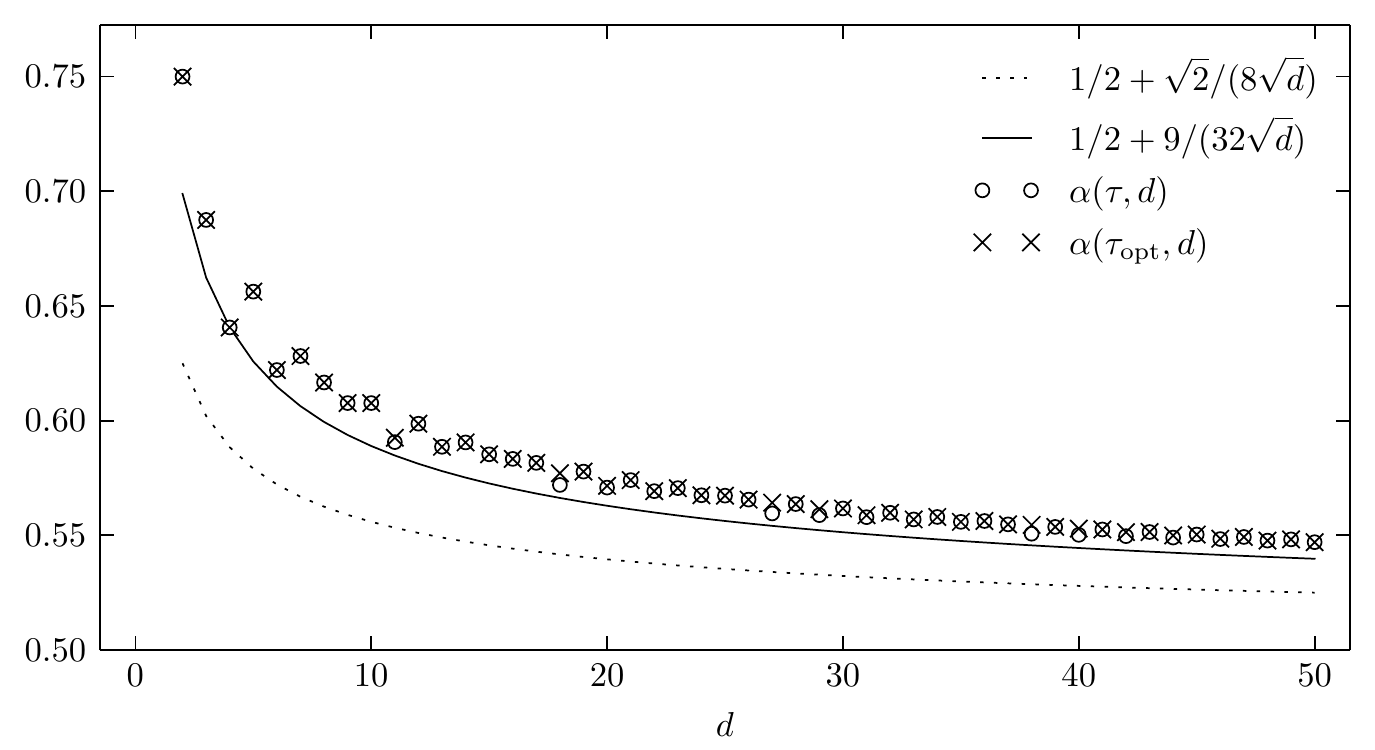}

    \medskip
    \small
    \begin{tabular}{@{}r@{ }l@{}}
    $\alpha(\tau_{\opt}, d)$: & expected weight of a cut found by an optimal threshold algorithm \\[1ex]
    $\alpha(\tau, d)$: & expected weight of a cut found by our algorithm \\[1ex]
    $\frac{1}{2} + \frac{9}{32 \sqrt{d}}$: & a lower bound on $\alpha(\tau, d)$ \\[1ex]
    $\frac{1}{2} + \frac{\sqrt{2}}{8 \sqrt{d}}$: & a lower bound for Shearer's~\cite{shearer92bipartite} algorithm \\
    \end{tabular}
    \bigskip
    \caption{A comparison of performance guarantees for $d = 2,3,\dots,50$. Note that our lower bound and the perfomance of the best threshold algorithm meet at $d = 4$. If we wanted to prove a general bound of the form $1/2 + C/\sqrt{d}$ for a larger $C > 9/32$, we would have to resort to more complicated algorithms (e.g., more than $1$ bit of randomness per node or more than $1$ communication round).}\label{fig:small}
\end{figure}

\section{Conclusions}\label{sec:concl}

In this work, we have presented a new randomised distributed algorithm for finding large cuts. The key observation was that the task of designing randomised distributed algorithms for finding large cuts can be reduced to the problem of finding a max-weight cut in a weighted neighbourhood graph. This way we were able to use computers to find optimal algorithms for small values of~$d$. The general form of the optimal algorithms was apparent, and hence the results were easy to generalise.

Our algorithm was designed for $d$-regular triangle-free graphs. However, it can be easily applied in a much more general setting as well. To see this, recall that $\alpha(\tau, d)$ is not only the expected weight of the cut, but it is also the probability that any individual edge $e = \{u,v\}$ is a cut edge. The analysis only assumes that $u$ and $v$ are of degree $d$ and they do not have a common neighbour. Hence we have the following immediate generalisations.
\begin{enumerate}
    \item Our algorithm can be applied in triangle-free graphs of \emph{maximum} degree $d$ as follows: a node of degree $d' < d$ simulates the behaviour of $d - d'$ missing neighbours. We still have the same guarantee that each original edge is a cut edge with probability $\alpha(\tau, d)$. The running time of the algorithm is still one communication round; however, some nodes need to produce more random bits.
    \item Our algorithm can also be applied in \emph{any} graph, even in those that contain triangles. Now our analysis shows that each edge that is not part of a triangle will be a cut edge with probability $\alpha(\tau,d)$. This observation already gives a simple bound: if at most a fraction $\epsilon$ of all edges are part of a triangle, we will find a cut of expected size at least ${(1-\epsilon)}\cdot\alpha(\tau,d)$.
\end{enumerate}

\section*{Acknowledgements}

Computer resources were provided by the Aalto University School of Science ``Science-IT'' project (Triton cluster), and by the Department of Computer Science at the University of Helsinki (Ukko cluster).

\clearpage

\def\UrlFont{\sf\footnotesize}
\bibliographystyle{plainnat}
\bibliography{local-maxcut}

\newpage
\appendix

\section{Proof of Theorem~\ref{thm:main}}\label{app:main}

We need to prove a lower bound on
\[
    \alpha(\tau, d)
    = \frac{1}{2} + \frac{1}{4^{d-1}}
      \binom{d-1}{\tau-1}
      \sum_{i=d-\tau+1}^{\tau-1} \binom{d-1}{i}
\]
in the region $\tau \approx d/2 + \sqrt{d}/2$. Our general strategy is as follows:
\begin{enumerate}[noitemsep]
    \item Verify cases $d = 2,3,\dotsc,3000$ with a computer.
    \item Prove a closed-form lower bound for $d > 3000$.
\end{enumerate}
The first part is easily solved with a simple Python script or with a short calculation in Mathematica (see Figure~\ref{fig:small} for examples of the results for $d = 2,3,\dotsc,50$). We will now focus on the second part; for that we will need various estimates of binomial coefficients.

The proof given here is certainly not the most elegant way to derive the bound, but it is self-contained and gets the job done. Proving the claim for a ``sufficiently large'' $d$ would be straightforward. However, we need to show that already a concrete relatively small $d$ such as $d > 3000$ is enough.

We will first approximate binomial coefficients with the normal distribution. Let $J = \{1,2,3,4\}$, and define
\[
    \delta_j(n) = \Bigl\lfloor j \sqrt{n/32} \Bigr\rfloor, \quad
    g_j = e^{-j^2/32}
\]
for each $j \in \{0\} \cup J$.

\begin{fact}\label{fact:center}
    For any $n \ge 1500$ we have
    \[
        \frac{0.999}{\sqrt{\pi n}} < \frac{1}{4^n} \binom{2n}{n} < \frac{1}{\sqrt{\pi n}}.
    \]
\end{fact}

\begin{lemma}\label{lem:offcenter}
    For any $j \in J$, $\delta = \delta_j(n)$, and $n \ge 1500$ we have
    \[
        \binom{2n}{n + \delta} > 0.995 \cdot g_j \cdot \binom{2n}{n}
    \]
\end{lemma}
\begin{proof}
    We can estimate
    \[
    \begin{split}
        \binom{2n}{n + \delta} / \binom{2n}{n}
        &= \frac{n!}{(n + \delta)!} \cdot \frac{n!}{(n - \delta)!} \\[1ex]
        &= \frac{n-\delta+1}{n+1} \cdot \frac{n-\delta+2}{n+2} \dotsm \frac{n}{n+\delta}
        > \left(1 - \frac{\delta}{n}\right)^\delta \ge h_j(\delta),
    \end{split}
    \]
    where
    \[
        h_j(\delta) = \left(1 - \frac{j^2}{32 \delta}\right)^\delta.
    \]
    Now $h_j(\delta) \to g_j$ as $\delta \to \infty$. For each $j \in J$ we can verify that $h_j(\delta) > 0.995 \cdot g_j$ when $\delta \ge \delta_j(1500)$.
\end{proof}

\begin{lemma}\label{lem:cumsum}
    For $\delta = \delta_4(n)$ and $n \ge 1500$ we have
    \begin{align*}
        \frac{1}{4^n} \sum_{i = -\delta+1}^{\delta} \binom{2n}{n + i} &> 0.6088, &
        \frac{1}{4^n} \sum_{i = -\delta+1}^{\delta-1} \binom{2n}{n + i} &> 0.5975.
    \end{align*}
\end{lemma}
\begin{proof}
    Here we could apply the Berry--Esseen theorem, but the following simple piecewise estimate is sufficient for our purposes. As
    \[
        \delta_j(n) > j \sqrt{n/32} - 1,
    \]
    we have
    \[
        \sum_{j = 1}^{4} \bigl(\delta_j(n) - \delta_{j-1}(n)\bigr) \cdot g_j
         > \biggl( \sum_{j = 1}^{4} g_j \sqrt{n/32} \biggr) - g_1
         > 0.5680 \sqrt{n} - 0.9693.
    \]
    Hence using Fact~\ref{fact:center} and Lemma~\ref{lem:offcenter} we have
    \[
    \begin{split}
        \frac{1}{4^n} \sum_{i = 1}^{\delta} \binom{2n}{n + i}
        &\ge \frac{1}{4^n} \sum_{j = 1}^{4} \bigl(\delta_j(n) - \delta_{j-1}(n)\bigr) \binom{2n}{n + \delta_j(n)} \\
        &\ge 0.995 \cdot \frac{1}{4^n} \binom{2n}{n} \sum_{j = 1}^{4} \bigl(\delta_j(n) - \delta_{j-1}(n)\bigr) g_j \\
        &> 0.995 \cdot \frac{0.999}{\sqrt{\pi n}} \cdot \bigl( 0.5680 \sqrt{n} - 0.9693 \bigr) \\
        &> 0.3185 - 0.5436 / \sqrt{n}
         > 0.3044.
    \end{split}
    \]
    The claim follows from the observations
    \begin{align*}
        \frac{1}{4^n} \sum_{i = -\delta + 1}^{\delta} \binom{2n}{n + i}
        &> \frac{2}{4^n} \sum_{i = 1}^{\delta} \binom{2n}{n + i}
         > 2 \cdot 0.3044 = 0.6088, \\
        \frac{1}{4^n} \sum_{i = -\delta + 1}^{\delta-1} \binom{2n}{n + i}
        &> \Bigl(2 - \frac{1}{\delta}\Bigr) \frac{1}{4^n} \sum_{i = 1}^{\delta} \binom{2n}{n + i}
         > 1.9629 \cdot 0.3044 > 0.5975.
        \qedhere
    \end{align*}
\end{proof}

Now we have the estimates that we will use in the proof of Theorem~\ref{thm:main}. We will consider the odd and even values of $d$ separately.

\paragraph{Odd {\boldmath $d$}.}

Assume that $d = 2n+1$, $n \ge 1500$. Let
\[
    \delta = \tau - n = \Big\lceil \sqrt{n/2+1/4} + 1/2 \Big\rceil, \quad \delta' = \delta_4(n),
\]
and observe that
\[
    \sqrt{n/2} < \sqrt{n/2+1/4} + 1/2 < \sqrt{n/2} + 1.
\]
It follows that
\[
    \delta' + 1 \le \delta \le \delta' + 2.
\]
Therefore
\[
\begin{split}
    \alpha(\tau, d)
    &= \frac{1}{2} + \frac{1}{4^{d-1}}
       \binom{d-1}{\tau-1}
       \sum_{i=d-\tau+1}^{\tau-1} \binom{d-1}{i} \\
    &= \frac{1}{2} + \frac{1}{4^{2n}}
       \binom{2n}{n+\delta-1}
       \sum_{i=-\delta+2}^{\delta-1} \binom{2n}{n+i} \\
    &\ge \frac{1}{2} + \frac{1}{4^{2n}}
       \binom{2n}{n+\delta'+1}
       \sum_{i=-\delta'+1}^{\delta'} \binom{2n}{n+i} \\
    &= \frac{1}{2} + 
       \frac{n-\delta'}{n+\delta'+1} \cdot
       \frac{1}{4^n} \binom{2n}{n+\delta'} \cdot
       \frac{1}{4^n} \sum_{i=-\delta'+1}^{\delta'} \binom{2n}{n+i} \\
    &> \frac{1}{2} + 
       0.964 \cdot
       0.995 \cdot g_4 \cdot \frac{0.999}{\sqrt{\pi n}} \cdot 0.6088
     > \frac{1}{2} + \frac{0.2823}{\sqrt{d-1}}
     > \frac{1}{2} + \frac{9}{32\sqrt{d}}.
\end{split}
\]

\paragraph{Even {\boldmath $d$}.}

Assume that $d = 2n$, $n > 1500$. Let
\[
    \delta = \tau - n = \Big\lceil \sqrt{n/2} \Big\rceil, \quad \delta' = \delta_4(n).
\]
Now we have
\[
    \delta' \le \delta \le \delta' + 1.
\]
For any $k < n$ we have the identity
\[
\begin{split}
       \sum_{i=-k}^{k} \binom{2n}{n+i}
    &= \sum_{i=-k}^{k} \biggl( \binom{2n-1}{n+i-1} + \binom{2n-1}{n+i} \biggr) \\
    &= \sum_{i=-k}^{k} \biggl( \binom{2n-1}{n-i} + \binom{2n-1}{n+i} \biggr)
     =2\sum_{i=-k}^{k} \binom{2n-1}{n+i}.
\end{split}
\]
We can use it to derive
\[
\begin{split}
    \alpha(\tau, d)
    &= \frac{1}{2} + \frac{1}{4^{d-1}}
       \binom{d-1}{\tau-1}
       \sum_{i=d-\tau+1}^{\tau-1} \binom{d-1}{i} \\
    &= \frac{1}{2} + \frac{1}{4^{2n-1}}
       \binom{2n-1}{n+\delta-1}
       \sum_{i=-\delta+1}^{\delta-1} \binom{2n-1}{n+i} \\
    &\ge \frac{1}{2} + \frac{1}{4^{2n-1}}
       \binom{2n-1}{n+\delta'}
       \sum_{i=-\delta'+1}^{\delta'-1} \binom{2n-1}{n+i} \\
    &= \frac{1}{2} + \frac{1}{4^{2n-1}}
       \cdot
       \frac{n-\delta'}{2n}
       \binom{2n}{n+\delta'}
       \cdot
       \frac{1}{2}
       \sum_{i=-\delta'+1}^{\delta'-1} \binom{2n}{n+i} \\
    &= \frac{1}{2} + 
       \frac{n-\delta'}{n}
       \cdot
       \frac{1}{4^n}
       \binom{2n}{n+\delta'}
       \cdot
       \frac{1}{4^n}
       \sum_{i=-\delta'+1}^{\delta'-1} \binom{2n}{n+i} \\
    &> \frac{1}{2} + 
       0.982 \cdot 0.995 \cdot g_4 \cdot \frac{0.999}{\sqrt{\pi n}} \cdot 0.5975
     > \frac{1}{2} + \frac{0.2822}{\sqrt{d}}
     > \frac{1}{2} + \frac{9}{32\sqrt{d}}.
\end{split}
\]
This completes the proof of Theorem~\ref{thm:main}.

\end{document}